\documentclass[12pt]{article}
\usepackage{amssymb,amsmath,amsthm,amsfonts,mathrsfs}
\usepackage{enumerate}

\theoremstyle{definition}

\usepackage{pgf,tikz}
\usetikzlibrary{intersections}
\usepackage{graphicx}
\usepackage{float}
\usepackage[colorlinks=true, allcolors=blue]{hyperref}
\usepackage[all]{hypcap}
\usepackage[capitalize,nameinlink]{cleveref}
\usepackage{titlesec}
\titlelabel{\thetitle.\quad}
\newtheorem{theo}{Theorem}

\theoremstyle{definition}

\newtheorem{definition}{Definition}
\date{}
\begin{document}
\title{Vertex Ordering Characterizations of Interval $r$-Graphs}
\author{ Indrajit Paul, Ashok Kumar Das\\Department of Pure Mathematics\\ University of Calcutta\\ Kolkata, West Bengal\\ India\\
Email Address -  
paulindrajit199822@gmail.com \&\\ ashokdas.cu@gmail.com}
\maketitle
\begin{abstract}
An $r$-partite graph is an interval $r$-graph if corresponding to each vertex we can assign an interval of the real line such that two vertices $u$ and $v$ of different partite sets are adjacent if and only if their corresponding intervals intersect. In this paper, we provide two vertex-ordering characterizations of interval $r$-graphs, utilizing generalized interval orderings and $r$-interval orderings, and identify forbidden patterns for interval $r$-graphs in terms of specific orderings of their vertices.
\end{abstract}

\section{Introduction}
A graph $G=(V,E)$ is called an interval graph, if there exists a family $\mathcal{I}=\{I_v:v\in V\}$, of intervals of the real line such that, for $a,b\in V$ the vertices $a$ and $b$ are adjacent in $G$ if and only if $I_a\cap I_b\neq\emptyset$. A bigraph $B=(X,Y,E)$ is said to be an interval bigraph if there exists a family $\mathcal{I}=\{I_v:v\in X\cup Y\}$ of intervals of the real line such that for all $x\in X$ and $y\in Y$ the vertices $x$ and $y$ are adjacent in $B$ if and only if $I_x\cap I_y\neq\emptyset$. Interval bigraphs were introduced in \cite{harary} and have been studied in \cite{brown,Ashok-sandip-malay,hell-huang,muller1}. Interval digraphs were introduced by Sen et al. \cite{sen-das-roy-west}. It was shown in \cite{Ashok-sandip-malay} that the notions of interval digraphs and interval bigraphs are equivalent. Interval bigraphs or interval digraphs have become of interest in such new areas as graph homomorphism, e.g \cite{feder-hell-huang-rafiey}. The concept of interval bigraph was generalized by Brown \cite{Brown}. A graph $G=(V,E)$ is called interval $r$-graph if the vertex set $V$ can be partitioned into $r$ stable sets say $V_1$, $V_2$,\dots, $V_r$ and corresponding to each vertex we can assign an interval such that two vertices $u$ and $v$ are adjacent in $G$ if and only if the corresponding intervals $I_u$ and $I_v$ have non-empty intersection and $u$ and $v$ belong to different partite sets of vertices.\\
In this paper, we study different types of vertex orderings of interval $r$-graphs and, with the help of these orderings, we characterize the class of interval $r$-graphs. Finally, we provide forbidden patterns for interval $r$-graphs with respect to a particular vertex ordering. 
\section{Main Result}
In this section, we introduce several types of vertex orderings for
$r$-partite graphs and using these orderings we characterize interval
$r$-graphs. Based on these orderings, we provide characterizations and identify forbidden patterns for interval $r$-graphs.\\
We begin with the definition of \textit{generalized interval ordering} of the vertices of an $r$-partite graph.\\
\begin{definition}\label{d1}
Let $G = (X_1, X_2, \ldots, X_r, E)$ be an $r$-partite graph of order $n$. 
Arrange the vertices of $G$ linearly and label them $v_1, v_2, \ldots, v_n$ such that 
the vertex $v_i$ is placed on the $i$-th positive integer of the real line. 
Let $X = \bigcup_{i=1}^{r} X_i$ denote the vertex set of $G$. 
The set $X$ is said to admit a \emph{generalized interval ordering} 
if the following condition is satisfied:
\begin{itemize}
\item For any edge $v_i v_j \in E$ with $i > j$, where $v_i \in X_\alpha$, 
$v_j \in X_\beta$, and $\alpha \ne \beta$, it holds that
$ v_\ell v_j \in E \quad \text{for all } v_\ell \notin X_\beta 
\text{ with } \ell \in \{ j + 1, j + 2, \ldots, i - 1 \}.$
\end{itemize}

\end{definition}
We shall now characterize interval $r$-graphs using this ordering.
\begin{theo}\label{t1}
    \textit{An $r$-partite graph $G=(X_1,X_2,...X_r,E)$ is a interval $r$-graph if and only if the vertex set $X=\bigcup_{i=1}^{r} X_i $ of $G$ admits a \textit{generalized interval ordering}}.
\end{theo}
\begin{proof}
    Necessity: Let $G = (X_1, X_2, \dots, X_r, E)$ be an \emph{interval $r$-graph}. Then there exists a class 
$\mathcal{I} = \{ I_v : v \in \bigcup_{i=1}^r X_i \}$ of intervals corresponding to each vertex $v$ of $G$, 
such that $uv \in E$ if and only if $I_u \cap I_v \neq \emptyset$, where $u$ and $v$ belong to different partite sets. 
Without loss of generality, we assume that all intervals have distinct endpoints. 

Now, order the vertices of $G$ as $v_1, v_2, \dots, v_n$ according to the increasing order of the left endpoints 
of their corresponding intervals, where $n$ denotes the order of the $r$-partite graph $G$. 
We will show that this ordering is a generalized interval ordering  for the $r$-partite graph $G$.

Suppose $v_i v_j \in E$, where $i > j$ (and $v_i \in X_\alpha$, $v_j \in X_\beta$, $\alpha \neq \beta$). 
According to our ordering, the left endpoint of $v_i$ must lie within the interval $I_{v_j}$ corresponding to $v_j$. 
From Figure~1, it follows that the left endpoints of the intervals 
$I_{v_{j+1}}, I_{v_{j+2}}, I_{v_{j+3}}, \dots, I_{v_{i-1}}$ are all contained in $I_{v_j}$. Thus $v_lv_j\in E$ if $v_l\notin X_\beta$ for $l\in \{j+1, \dots, i-1\}$. 
Hence, by Definition~1, we conclude that the ordering 
$v_1, v_2, \dots, v_n$ of the vertices of $G$ is a generalized interval ordering.
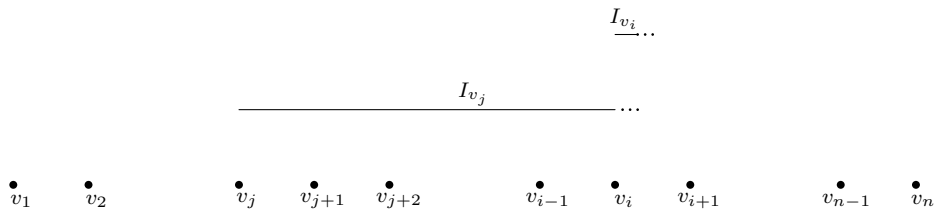
\begin{figure}[H]
    \centering
    \begin{tikzpicture}[line cap=round,line join=round,x=1.0cm,y=1.0cm]
\clip(-1,-0.5) rectangle (13,2.5);
\draw (3,1)-- (8,1);
\fill [color=black] (8.2,1) circle (.5pt);
\fill [color=black] (8.1,1) circle (.5pt);
\fill [color=black] (8.3,1) circle (.5pt);

\draw (8,2)-- (8.3,2);
\begin{scriptsize}
\fill [color=black] (0,0) circle (1.5pt);
\draw[color=black] (0.11,-0.21) node {$v_1$};
\fill [color=black] (1,0) circle (1.5pt);
\draw[color=black] (1.11,-0.21) node {$v_2$};
\fill [color=black] (3,0) circle (1.5pt);
\draw[color=black] (3.11,-0.21) node {$v_j$};
\fill [color=black] (4,0) circle (1.5pt);
\draw[color=black] (4.12,-0.21) node {$v_{j+1}$};
\fill [color=black] (5,0) circle (1.5pt);
\draw[color=black] (5.13,-0.21) node {$v_{j+2}$};
\fill [color=black] (7,0) circle (1.5pt);
\draw[color=black] (7.11,-0.21) node {$v_{i-1}$};
\fill [color=black] (8,0) circle (1.5pt);
\draw[color=black] (8.12,-0.21) node {$v_i$};
\fill [color=black] (9,0) circle (1.5pt);
\draw[color=black] (9.12,-0.21) node {$v_{i+1}$};
\fill [color=black] (11,0) circle (1.5pt);
\draw[color=black] (11.08,-0.21) node {$v_{n-1}$};
\fill [color=black] (12,0) circle (1.5pt);
\draw[color=black] (12.1,-0.21) node {$v_n$};
\fill [color=black] (8.4,2) circle (.5pt);
\fill [color=black] (8.3,2) circle (.5pt);
\fill [color=black] (8.5,2) circle (.5pt);
\draw[color=black] (8.12,2.21) node {$I_{v_i}$};
\draw[color=black] (6.12,1.21) node {$I_{v_j}$};

\end{scriptsize}
\end{tikzpicture}
    \caption{Left end point of $I_{v_i}$ is contained in $I_{v_j}$.}
    \label{fig:1}
\end{figure}
\textbf{Sufficiency:}  
Let $G = (X_1, X_2, \dots, X_r, E)$ be an $r$-partite graph and the ordering $v_1,v_2,\dots, v_n$  of its vertices is a generalized interval ordering.\\ Now, let
$m_i = \max \left( \{\, j : v_i v_j \in E \,\} \cup \{ i \} \right)$.  
Construct an interval $I'_{v_i}$ for each vertex $v_i$ $(1 \leq i \leq n)$ as  $I'_{v_i} = \left[\, i, m_i+ \frac{ 1}{2} \,\right]$.  First we shall show that the collection $\mathcal{I}' = \{\, I'_{v_i} : 1 \leq i \leq n \,\}$ of intervals 
satisfies the necessary conditions for $G$ to be an interval $r$-graph.  

Let $v_i$ and $v_j$ be two vertices belonging to different partite sets of $G$ such that 
$v_i v_j \in E$ and $i > j$.  
According to our construction,  
$I'_{v_i} = \left[\, i, m_i+\frac{1}{2} \,\right]
\quad \text{and} \quad
I'_{v_j} = \left[\, j, m_j+\frac{1}{2} \,\right]$. 
Since $v_i v_j \in E$ and $i > j$, we have $m_j \ge i$, which implies 
$m_j+\frac{ 1}{2} \ge i$. 
Therefore, the interval $I'_{v_j}$ contains $i$, and hence 
$I'_{v_i} \cap I'_{v_j} \neq \emptyset$.  

Conversely, let $v_i$ and $v_j$ be two vertices of $G$ from different partite sets such that 
$I'_{v_i} \cap I'_{v_j} \neq \emptyset$ and $i > j$.  
Then $i \in I'_{v_j}$, which implies $m_j \ge i$.  
Thus, $j < i \le m_j$, and since $v_j v_{m_j} \in E$, by the definition of a generalized interval 
ordering we have $v_i v_j \in E$. Hence,  
$v_i v_j \in E \quad \Longleftrightarrow \quad I'_{v_i} \cap I'_{v_j} \neq \emptyset$.
Therefore, $G$ is an interval $r$-graph.

\end{proof}
Before presenting another vertex-ordering characterization, we define the concept of almost consecutive ones in the rows and columns of the adjacency matrix of $r$-partite graphs.

\begin{definition}\label{d2}
    Let $G = (X_1, X_2, \ldots, X_r, E)$ be an $r$-partite graph and let $\textbf{A}$ be the adjacency matrix of $G$. A row (say, the $i$-th row) is said to have \emph{almost consecutive ones} if, between any two ones in that row, whenever a zero appears at the position $(v_i, v_k)$, the vertices $v_i$ and $v_k$ belong to the same partite set of $G$.
\end{definition} 
\begin{definition}\label{d3} 
Let $G = (X_1, X_2, \ldots, X_r, E)$ be an $r$-partite graph and let $\textbf{A}$ be the adjacency matrix of $G$. A column (say, the $j$-th column) is said to have \emph{almost consecutive ones} if, between any two ones in that column, whenever a zero appears at the position $(v_l, v_j)$, the vertices $v_l$ and $v_j$ belong to the same partite set of $G$.
\end{definition}
\par We now introduce another type of vertex ordering for $r$-partite graphs, referred to as the 
\emph{$r$-interval ordering}. Using this ordering, we characterize the class of 
\emph{interval $r$-graphs}.  

Let $G = (X_1, X_2, \dots, X_r, E)$ be an $r$-partite graph of order $n$, and let 
$v_1, v_2, \dots, v_n$ be an ordering of its vertices. 
Place these vertices on the real line such that the $i$-th vertex $v_i$ is placed on the positive integer point $i$.  

Let $\textbf{A}$ denote the adjacency matrix of $G$, where both rows and columns are arranged in the increasing order of vertex indices. 
Consider any row of $\textbf{A}$ (say, the $i$-th row). Define $R_i$ as the set of entries of $1$ in this row that appear \emph{almost consecutively}, 
beginning with the column corresponding to the vertex $v_{s_i}$, where $v_{s_i}$ is the first vertex to the right of $v_i$ that belongs to a different partite set and $s_i>i$.  
If no $v_i v_{s_i} \in E$, then $R_i = \emptyset$. Otherwise, the sequence continues rightward until the last $1$ in this \emph{almost consecutive} manner is reached.  

Similarly, consider any column of $\textbf{A}$ (say, the $j$-th column). Define $C_j$ as the set of entries of $1$ in this column that appear 
\emph{almost consecutively}, beginning with the row corresponding to the vertex $v_{s_j}$, where $v_{s_j}$ is the first vertex to the right of $v_j$ 
that belongs to a different partite set and $s_j>j$.  
If no $v_j v_{s_j} \in E$, then $C_j = \emptyset$. Otherwise, the sequence continues downward until the last $1$ in this \emph{almost consecutive} manner is reached.  

An ordering $v_1, v_2, \dots, v_n$ of the vertices of $G$ is called an \emph{$r$-interval ordering} 
if the sets $R_i$ and $C_j$ together contain all the $1$’s in the adjacency matrix $\textbf{A}$.

\begin{theo}\label{t2}
    An $r$-partite graph $G = (X_1, X_2, \dots, X_r, E)$ is an interval $r$-graph if and only if its vertex set $ X = \bigcup_{i=1}^r X_i $ admits an $r$-interval ordering.
\end{theo}
\begin{proof}
Necessity: Let \( G = (X_1, X_2, \ldots, X_r, E) \) be an interval \( r \)-graph of order \( n \).
Then there exists an interval model
\[
\mathcal{I} = \{ I_v : v \in X = \bigcup_{i=1}^r X_i \}
\]
such that for any two vertices \( u, v \) belonging to different partite sets,
\[
uv \in E \quad \text{if and only if} \quad I_u \cap I_v \neq \emptyset.
\]

Without loss of generality, we may assume that the interval model \( \mathcal{I} \) satisfies the following conditions:
\begin{enumerate}
    \item All intervals are closed, i.e., each interval includes its endpoints.
    \item No two intervals share the same left endpoint.
\end{enumerate}

Label the vertices of \( G \) as \( v_1, v_2, \ldots, v_n \) according to the \textit{increasing order of the left endpoints} of their corresponding intervals.
Arrange both the rows and columns of the adjacency matrix \( A \) of \( G \) according to this vertex ordering.

We claim that, under this arrangement, the sets \( R_i \) and the sets \( C_j \) together contain all the 1’s of the adjacency matrix \( A \).

\medskip
Suppose an entry of $A$ is \( a_{ij} = 1 \).
This implies that \( v_i \) and \( v_j \) are adjacent in \( G \), i.e., \( v_i v_j \in E \).
Let \( v_i \in X_\alpha \) and \( v_j \in X_\beta \), where \( \alpha \neq \beta \).
We distinguish two cases according to the relative positions of \( v_i \) and \( v_j \) in the vertex ordering.

\paragraph{Case 1:} When \( i > j \).
From the interval representation (see Figure~2), any vertex appearing between \( v_j \) and \( v_i \) in this ordering, and belonging to a partite set other than \( X_\beta \), will be adjacent to \( v_j \).
Consequently, \( C_j \) will contain the 1 corresponding to the entry \( (i, j) \) in the adjacency matrix.
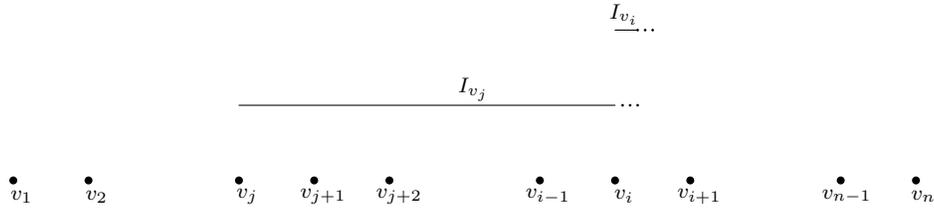
\begin{figure}[H]
    \centering
    \begin{tikzpicture}[line cap=round,line join=round,x=1.0cm,y=1.0cm]
\clip(-1,-0.5) rectangle (13,2.5);
\draw (3,1)-- (8,1);
\fill [color=black] (8.2,1) circle (.5pt);
\fill [color=black] (8.1,1) circle (.5pt);
\fill [color=black] (8.3,1) circle (.5pt);

\draw (8,2)-- (8.3,2);
\begin{scriptsize}
\fill [color=black] (0,0) circle (1.5pt);
\draw[color=black] (0.11,-0.21) node {$v_1$};
\fill [color=black] (1,0) circle (1.5pt);
\draw[color=black] (1.11,-0.21) node {$v_2$};
\fill [color=black] (3,0) circle (1.5pt);
\draw[color=black] (3.11,-0.21) node {$v_j$};
\fill [color=black] (4,0) circle (1.5pt);
\draw[color=black] (4.12,-0.21) node {$v_{j+1}$};
\fill [color=black] (5,0) circle (1.5pt);
\draw[color=black] (5.13,-0.21) node {$v_{j+2}$};
\fill [color=black] (7,0) circle (1.5pt);
\draw[color=black] (7.11,-0.21) node {$v_{i-1}$};
\fill [color=black] (8,0) circle (1.5pt);
\draw[color=black] (8.12,-0.21) node {$v_i$};
\fill [color=black] (9,0) circle (1.5pt);
\draw[color=black] (9.12,-0.21) node {$v_{i+1}$};
\fill [color=black] (11,0) circle (1.5pt);
\draw[color=black] (11.08,-0.21) node {$v_{n-1}$};
\fill [color=black] (12,0) circle (1.5pt);
\draw[color=black] (12.1,-0.21) node {$v_n$};
\fill [color=black] (8.4,2) circle (.5pt);
\fill [color=black] (8.3,2) circle (.5pt);
\fill [color=black] (8.5,2) circle (.5pt);
\draw[color=black] (8.12,2.21) node {$I_{v_i}$};
\draw[color=black] (6.12,1.21) node {$I_{v_j}$};

\end{scriptsize}
\end{tikzpicture}
    \caption{$I_{v_i}\cap I_{v_j}\neq\emptyset$, where $i>j$.}
    \label{fig:2}
\end{figure}

\paragraph{Case 2:} When \( i < j \).
Similarly, any vertex appearing between \( v_i \) and \( v_j \) that does not belong to the same partite set \( X_\alpha \) will be adjacent to \( v_i \) (see Figure~3).
Hence, \( R_i \) will contain the 1 corresponding to the entry \( (i, j) \) in the adjacency matrix.
\begin{figure}[H]
    \centering
    \begin{tikzpicture}[line cap=round,line join=round,x=1.0cm,y=1.0cm]
\clip(-1,-0.5) rectangle (13,2.5);
\draw (3,1)-- (8,1);
\fill [color=black] (8.2,1) circle (.5pt);
\fill [color=black] (8.1,1) circle (.5pt);
\fill [color=black] (8.3,1) circle (.5pt);

\draw (8,2)-- (8.3,2);
\begin{scriptsize}
\fill [color=black] (0,0) circle (1.5pt);
\draw[color=black] (0.11,-0.21) node {$v_1$};
\fill [color=black] (1,0) circle (1.5pt);
\draw[color=black] (1.11,-0.21) node {$v_2$};
\fill [color=black] (3,0) circle (1.5pt);
\draw[color=black] (3.11,-0.21) node {$v_i$};
\fill [color=black] (4,0) circle (1.5pt);
\draw[color=black] (4.12,-0.21) node {$v_{i+1}$};
\fill [color=black] (5,0) circle (1.5pt);
\draw[color=black] (5.13,-0.21) node {$v_{i+2}$};
\fill [color=black] (7,0) circle (1.5pt);
\draw[color=black] (7.11,-0.21) node {$v_{j-1}$};
\fill [color=black] (8,0) circle (1.5pt);
\draw[color=black] (8.12,-0.21) node {$v_j$};
\fill [color=black] (9,0) circle (1.5pt);
\draw[color=black] (9.12,-0.21) node {$v_{j+1}$};
\fill [color=black] (11,0) circle (1.5pt);
\draw[color=black] (11.08,-0.21) node {$v_{n-1}$};
\fill [color=black] (12,0) circle (1.5pt);
\draw[color=black] (12.1,-0.21) node {$v_n$};
\fill [color=black] (8.4,2) circle (.5pt);
\fill [color=black] (8.3,2) circle (.5pt);
\fill [color=black] (8.5,2) circle (.5pt);
\draw[color=black] (8.12,2.21) node {$I_{v_j}$};
\draw[color=black] (6.12,1.21) node {$I_{v_i}$};

\end{scriptsize}
\end{tikzpicture}
    \caption{$I_{v_i}\cap I_{v_j}\neq\emptyset$, where $i<j$.}
    \label{fig:3}
\end{figure}
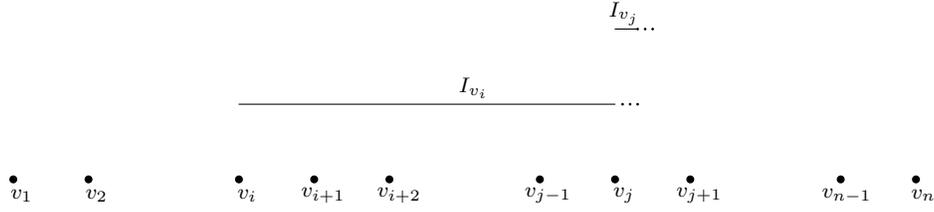
\medskip
In either case, the 1 at the \( (i, j) \)-th entry is contained in one of \( R_i \) or \( C_j \).
Therefore, the sets \( R_i \) and \( C_j \) together include all the 1’s of the adjacency matrix \( A \).\vspace{.3cm}\\
Sufficiency:
Consider an \( r \)-partite graph \( G = (X_1, X_2, \ldots, X_r, E) \), where the vertex set
$X = \bigcup_{i=1}^{r} X_i$
is ordered as \( v_1, v_2, \ldots, v_n \).
Assume that this ordering is an \textit{\( r \)-interval ordering}.

Assume if the rows and columns of the adjacency matrix of \( G \) are arranged according to this ordering,
then the sets \( R_i \) and \( C_j \) together contain all the 1’s in the adjacency matrix.\\
We now construct an interval model for \( G \).
If \( R_i = \emptyset \), define \( I_{v_i} = [i, i] \).
Otherwise, let \( r_i \) be the greatest column index such that the entry at the \( (i, r_i) \)-th position of \( A \) is 1.
Then define \( I_{v_i} = [i, r_i] \).

It remains to show that the family of intervals
$\mathcal{I} = \{ I_{v_i} : 1 \leq i \leq n \}$
serves as an interval model for the \( r \)-partite graph \( G \).\\
Suppose $I_{v_i}\cap I_{v_j}\neq\emptyset$, where ($i<j$) and $v_i$, $v_j$ belong to different partite sets. Then by the construction of the intervals we can say that $j\in I_{v_i}$ $\implies \mathcal{R}_i$ contains the $1$ at position $(i,j)$ of the adjacency matrix of $G$. Therefore $v_i$ is adjacent to $v_j$.\\ 
Conversely, let $v_i$ is adjacent to $v_j$ in $G$, where $(i>j)$ , say. Therefore $\mathcal{R}_j$ will contain the 1 at the $(j,i)$-th entry. Implies $I_{v_j}$ contains $i$, hence by the definition of $I_{v_i}$ we have $I_{v_i}\cap I_{v_j}\neq\emptyset$. Therefore the class of intervals as constructed above is an interval model for the $r$-partite graph and so $G=(X_1,X_2,...,X_r,E)$ is an interval $r$-graph.

\end{proof}
\noindent\textbf{Example 1.} Consider the following $3$-partite graph in Figure~4. Using the above theorem we show that it is an interval $3$-graph. 
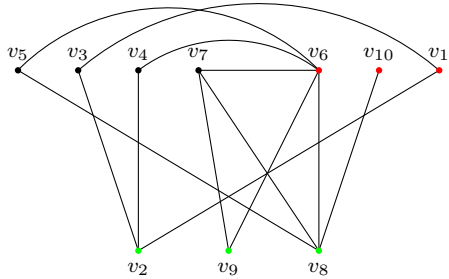
\begin{figure}[H]
    \centering
    \begin{tikzpicture}[line cap=round,line join=round,x=1.0cm,y=1.0cm,scale=.8]
\clip(2,2) rectangle (11,8);
\draw [shift={(5.5,3.5)}] plot[domain=0.79:2.36,variable=\t]({1*3.54*cos(\t r)+0*3.54*sin(\t r)},{0*3.54*cos(\t r)+1*3.54*sin(\t r)});
\draw [shift={(7,2.5)}] plot[domain=0.86:2.28,variable=\t]({1*4.61*cos(\t r)+0*4.61*sin(\t r)},{0*4.61*cos(\t r)+1*4.61*sin(\t r)});

\draw [shift={(6.5,4)}] plot[domain=0.93:2.21,variable=\t]({1*2.5*cos(\t r)+0*2.5*sin(\t r)},{0*2.5*cos(\t r)+1*2.5*sin(\t r)});
\draw (6,6)-- (8,6);
\draw (3,6)-- (8,3);
\draw (4,6)-- (5,3);
\draw (5,6)-- (5,3);
\draw (5,3)-- (10,6);
\draw (6,6)-- (6.5,3);
\draw (6,6)-- (8,3);
\draw (8,6)-- (6.5,3);
\draw (8,6)-- (8,3);
\draw (9,6)-- (8,3);
\begin{scriptsize}
\fill [color=black] (6,6) circle (1.5pt);
\draw[color=black] (6.,6.27) node {$v_7$};

\fill [color=black] (3,6) circle (1.5pt);
\draw[color=black] (3.,6.27) node {$v_5$};
\fill [color=black] (4,6) circle (1.5pt);
\draw[color=black] (4.,6.27) node {$v_3$};
\fill [color=black] (5,6) circle (1.5pt);
\draw[color=black] (5.,6.27) node {$v_4$};
\fill [color=red] (8,6) circle (1.5pt);
\draw[color=black] (8.,6.27) node {$v_6$};
\fill [color=red] (9,6) circle (1.5pt);
\draw[color=black] (9.,6.27) node {$v_{10}$};
\fill [color=red] (10,6) circle (1.5pt);
\draw[color=black] (10.,6.27) node {$v_1$};
\fill [color=green] (5,3) circle (1.5pt);
\draw[color=black] (5.,2.7) node {$v_2$};
\fill [color=green] (6.5,3) circle (1.5pt);
\draw[color=black] (6.5,2.7) node {$v_9$};
\fill [color=green] (8,3) circle (1.5pt);
\draw[color=black] (8.,2.7) node {$v_8$};
\end{scriptsize}
\end{tikzpicture}
    \caption{A $3$-partite graph.}
    \label{fig:4}
\end{figure}

\begin{figure}[H]
    \centering

    \begin{tikzpicture}[line cap=round,line join=round,x=1cm,y=1cm,scale=1.5]
\clip(.5,.9) rectangle (11.2,11.5);
\draw [line width=1.pt] (1,1)-- (11.,1);
\draw [line width=1.pt] (1,1)-- (1.,11);
\draw [line width=1.pt] (1,11)-- (11.,11);
\draw [line width=1.pt] (11,1)-- (11.,11);
\draw [line width=1.pt] (2,1)-- (2.,11);
\draw [line width=1.pt] (3,1)-- (3.,11);
\draw [line width=1.pt] (4,1)-- (4.,11);
\draw [line width=1.pt] (5,1)-- (5.,11);
\draw [line width=1.pt] (6,1)-- (6.,11);
\draw [line width=1.pt] (7,1)-- (7.,11);
\draw [line width=1.pt] (8,1)-- (8.,11);
\draw [line width=1.pt] (9,1)-- (9.,11);
\draw [line width=1.pt] (10,1)-- (10.,11);
\draw [line width=1.pt] (1,2)-- (11.,2);
\draw [line width=1.pt] (1,3)-- (11.,3);
\draw [line width=1.pt] (1,4)-- (11.,4);
\draw [line width=1.pt] (1,5)-- (11.,5);
\draw [line width=1.pt] (1,6)-- (11.,6);
\draw [line width=1.pt] (1,7)-- (11.,7);
\draw [line width=1.pt] (1,8)-- (11.,8);
\draw [line width=1.pt] (1,9)-- (11.,9);
\draw [line width=1.pt] (1,10)-- (11.,10);

\draw (1.3,1.7) node[anchor=north west] {$0$};
\draw (2.3,1.7) node[anchor=north west] {$0$};
\draw (3.3,1.7) node[anchor=north west] {$0$};
\draw (4.3,1.7) node[anchor=north west] {$0$};
\draw (5.3,1.7) node[anchor=north west] {$0$};
\draw (6.3,1.7) node[anchor=north west] {$0$};
\draw (7.3,1.7) node[anchor=north west] {$0$};
\draw (8.3,1.7) node[anchor=north west] {$1$};
\draw (9.3,1.7) node[anchor=north west] {$0$};
\draw (10.3,1.7) node[anchor=north west] {$0$};

\draw (1.3,2.7) node[anchor=north west] {$0$};
\draw (2.3,2.7) node[anchor=north west] {$0$};
\draw (3.3,2.7) node[anchor=north west] {$0$};
\draw (4.3,2.7) node[anchor=north west] {$0$};
\draw (5.3,2.7) node[anchor=north west] {$0$};
\draw (6.3,2.7) node[anchor=north west] {$1$};
\draw (7.3,2.7) node[anchor=north west] {$1$};
\draw (8.3,2.7) node[anchor=north west] {$0$};
\draw (9.3,2.7) node[anchor=north west] {$0$};
\draw (10.3,2.7) node[anchor=north west] {$0$};

\draw (1.3,3.7) node[anchor=north west] {$0$};
\draw (2.3,3.7) node[anchor=north west] {$0$};
\draw (3.3,3.7) node[anchor=north west] {$0$};
\draw (4.3,3.7) node[anchor=north west] {$0$};
\draw (5.3,3.7) node[anchor=north west] {$1$};
\draw (6.3,3.7) node[anchor=north west] {$1$};
\draw (7.3,3.7) node[anchor=north west] {$1$};
\draw (8.3,3.7) node[anchor=north west] {$0$};
\draw (9.3,3.7) node[anchor=north west] {$0$};
\draw (10.3,3.7) node[anchor=north west] {$1$};

\draw (1.3,4.7) node[anchor=north west] {$0$};
\draw (2.3,4.7) node[anchor=north west] {$0$};
\draw (3.3,4.7) node[anchor=north west] {$0$};
\draw (4.3,4.7) node[anchor=north west] {$0$};
\draw (5.3,4.7) node[anchor=north west] {$0$};
\draw (6.3,4.7) node[anchor=north west] {$1$};
\draw (7.3,4.7) node[anchor=north west] {$0$};
\draw (8.3,4.7) node[anchor=north west] {$1$};
\draw (9.3,4.7) node[anchor=north west] {$1$};
\draw (10.3,4.7) node[anchor=north west] {$0$};

\draw (1.3,5.7) node[anchor=north west] {$0$};
\draw (2.3,5.7) node[anchor=north west] {$0$};
\draw (3.3,5.7) node[anchor=north west] {$0$};
\draw (4.3,5.7) node[anchor=north west] {$1$};
\draw (5.3,5.7) node[anchor=north west] {$1$};
\draw (6.3,5.7) node[anchor=north west] {$0$};
\draw (7.3,5.7) node[anchor=north west] {$1$};
\draw (8.3,5.7) node[anchor=north west] {$1$};
\draw (9.3,5.7) node[anchor=north west] {$1$};
\draw (10.3,5.7) node[anchor=north west] {$0$};

\draw (1.3,6.7) node[anchor=north west] {$0$};
\draw (2.3,6.7) node[anchor=north west] {$0$};
\draw (3.3,6.7) node[anchor=north west] {$0$};
\draw (4.3,6.7) node[anchor=north west] {$0$};
\draw (5.3,6.7) node[anchor=north west] {$0$};
\draw (6.3,6.7) node[anchor=north west] {$1$};
\draw (7.3,6.7) node[anchor=north west] {$0$};
\draw (8.3,6.7) node[anchor=north west] {$1$};
\draw (9.3,6.7) node[anchor=north west] {$0$};
\draw (10.3,6.7) node[anchor=north west] {$0$};

\draw (1.3,7.7) node[anchor=north west] {$0$};
\draw (2.3,7.7) node[anchor=north west] {$1$};
\draw (3.3,7.7) node[anchor=north west] {$0$};
\draw (4.3,7.7) node[anchor=north west] {$0$};
\draw (5.3,7.7) node[anchor=north west] {$0$};
\draw (6.3,7.7) node[anchor=north west] {$1$};
\draw (7.3,7.7) node[anchor=north west] {$0$};
\draw (8.3,7.7) node[anchor=north west] {$0$};
\draw (9.3,7.7) node[anchor=north west] {$0$};
\draw (10.3,7.7) node[anchor=north west] {$0$};

\draw (1.3,8.7) node[anchor=north west] {$1$};
\draw (2.3,8.7) node[anchor=north west] {$1$};
\draw (3.3,8.7) node[anchor=north west] {$0$};
\draw (4.3,8.7) node[anchor=north west] {$0$};
\draw (5.3,8.7) node[anchor=north west] {$0$};
\draw (6.3,8.7) node[anchor=north west] {$0$};
\draw (7.3,8.7) node[anchor=north west] {$0$};
\draw (8.3,8.7) node[anchor=north west] {$0$};
\draw (9.3,8.7) node[anchor=north west] {$0$};
\draw (10.3,8.7) node[anchor=north west] {$0$};

\draw (1.3,9.7) node[anchor=north west] {$1$};
\draw (2.3,9.7) node[anchor=north west] {$0$};
\draw (3.3,9.7) node[anchor=north west] {$1$};
\draw (4.3,9.7) node[anchor=north west] {$1$};
\draw (5.3,9.7) node[anchor=north west] {$0$};
\draw (6.3,9.7) node[anchor=north west] {$0$};
\draw (7.3,9.7) node[anchor=north west] {$0$};
\draw (8.3,9.7) node[anchor=north west] {$0$};
\draw (9.3,9.7) node[anchor=north west] {$0$};
\draw (10.3,9.7) node[anchor=north west] {$0$};

\draw (1.3,10.7) node[anchor=north west] {$0$};
\draw (2.3,10.7) node[anchor=north west] {$1$};
\draw (3.3,10.7) node[anchor=north west] {$1$};
\draw (4.3,10.7) node[anchor=north west] {$0$};
\draw (5.3,10.7) node[anchor=north west] {$0$};
\draw (6.3,10.7) node[anchor=north west] {$0$};
\draw (7.3,10.7) node[anchor=north west] {$0$};
\draw (8.3,10.7) node[anchor=north west] {$0$};
\draw (9.3,10.7) node[anchor=north west] {$0$};
\draw (10.3,10.7) node[anchor=north west] {$0$};

\draw (1.3,11.5) node[anchor=north west] {$v_1$};
\draw (2.3,11.5) node[anchor=north west] {$v_2$};
\draw (3.3,11.5) node[anchor=north west] {$v_3$};
\draw (4.3,11.5) node[anchor=north west] {$v_4$};
\draw (5.3,11.5) node[anchor=north west] {$v_5$};
\draw (6.3,11.5) node[anchor=north west] {$v_6$};
\draw (7.3,11.5) node[anchor=north west] {$v_7$};
\draw (8.3,11.5) node[anchor=north west] {$v_8$};
\draw (9.3,11.5) node[anchor=north west] {$v_9$};
\draw (10.3,11.5) node[anchor=north west] {$v_{10}$};

\draw (0.5,10.7) node[anchor=north west] {$v_1$};
\draw (0.5,9.7) node[anchor=north west] {$v_2$};
\draw (0.5,8.7) node[anchor=north west] {$v_3$};
\draw (0.5,7.7) node[anchor=north west] {$v_4$};
\draw (0.5,6.7) node[anchor=north west] {$v_5$};
\draw (0.5,5.7) node[anchor=north west] {$v_6$};
\draw (0.5,4.7) node[anchor=north west] {$v_7$};
\draw (0.5,3.7) node[anchor=north west] {$v_8$};
\draw (0.5,2.7) node[anchor=north west] {$v_9$};
\draw (0.5,1.7) node[anchor=north west] {$v_{10}$};

\draw [rotate around={0.:(3,10.5)},line width=1.pt] (3,10.5) ellipse (1cm and 0.2cm);
\draw [rotate around={-90.:(1.5,9.)},line width=1.pt] (1.5,9.) ellipse (1cm and 0.2cm);
\draw [rotate around={0.:(4.,9.5)},line width=1.pt] (4.,9.5) ellipse (1cm and 0.2cm);
\draw [rotate around={-90.:(2.5,8)},line width=1.pt] (2.5,8) ellipse (1.cm and 0.2cm);
\draw [rotate around={0.:(6.5,7.5)},line width=1.pt] (6.5,7.5) ellipse (.5cm and 0.2cm);
\draw [rotate around={0.:(7.5,6.5)},line width=1.pt] (7.5,6.5) ellipse (1.5cm and 0.15cm);
\draw [rotate around={-90.:(4.5,5.5)},line width=1.pt] (4.5,5.5) ellipse (.5cm and 0.2cm);
\draw [rotate around={90.:(5.5,4.5)},line width=1.pt] (5.5,4.5) ellipse (1.5cm and 0.15cm);
\draw [rotate around={0.:(8.5,5.5)},line width=1.pt] (8.5,5.5) ellipse (1.5cm and 0.15cm);
\draw [rotate around={0.:(9,4.5)},line width=1.pt] (9,4.5) ellipse (1cm and 0.2cm);
\draw [rotate around={90.:(7.5,3)},line width=1.pt] (7.5,3) ellipse (1cm and 0.2cm);
\draw [rotate around={0.:(10.5,3.5)},line width=1.pt] (10.5,3.5) ellipse (.5cm and 0.15cm);
\draw [rotate around={90.:(6.5,3.5)},line width=1.pt] (6.5,3.5) ellipse (1.5cm and 0.2cm);
\draw [rotate around={90.:(8.5,1.5)},line width=1.pt] (8.5,1.5) ellipse (.5cm and 0.15cm);

\draw (2.2,11) node[anchor=north west,scale=.8] {$\mathcal{R}_1$};
\draw (2.2,10.4) node[anchor=north west,scale=1.5] {$\rightarrow$};
\draw (3.2,10.) node[anchor=north west,scale=.8] {$\mathcal{R}_2$};
\draw (3.2,9.4) node[anchor=north west,scale=1.5] {$\rightarrow$};
\draw (1.,9.7) node[anchor=north west,scale=.8] {$\mathcal{C}_1$};
\draw (1.5,9.7) node[anchor=north west,scale=1.5] {$\downarrow$};
\draw (2.,8.8) node[anchor=north west,scale=.8] {$\mathcal{C}_2$};
\draw (2.6,8.8) node[anchor=north west,scale=1.5] {$\downarrow$};
\draw (6,8) node[anchor=north west,scale=.8] {$\mathcal{R}_4$};
\draw (6,7.4) node[anchor=north west,scale=1.5] {$\rightarrow$};
\draw (4.,6) node[anchor=north west,scale=.8] {$\mathcal{C}_4$};
\draw (4.55,6) node[anchor=north west,scale=1.5] {$\downarrow$};
\draw (5.,6) node[anchor=north west,scale=.8] {$\mathcal{C}_5$};
\draw (5.5,6) node[anchor=north west,scale=1.5] {$\downarrow$};
\draw (6.2,7) node[anchor=north west,scale=.8] {$\mathcal{R}_5$};
\draw (6.2,6.4) node[anchor=north west,scale=1.5] {$\rightarrow$};
\draw (7.2,6) node[anchor=north west,scale=.8] {$\mathcal{R}_6$};
\draw (7.2,5.4) node[anchor=north west,scale=1.5] {$\rightarrow$};
\draw (8.2,5) node[anchor=north west,scale=.8] {$\mathcal{R}_7$};
\draw (8.2,4.4) node[anchor=north west,scale=1.5] {$\rightarrow$};
\draw (10.2,4) node[anchor=north west,scale=.8] {$\mathcal{R}_8$};
\draw (10.2,3.4) node[anchor=north west,scale=1.5] {$\rightarrow$};
\draw (6.,4.9) node[anchor=north west,scale=.8] {$\mathcal{C}_6$};
\draw (6.6,4.9) node[anchor=north west,scale=1.5] {$\downarrow$};
\draw (7,3.9) node[anchor=north west,scale=.8] {$\mathcal{C}_7$};
\draw (7.6,3.9) node[anchor=north west,scale=1.5] {$\downarrow$};
\draw (8,1.9) node[anchor=north west,scale=.8] {$\mathcal{C}_8$};
\draw (8.6,1.9) node[anchor=north west,scale=1.5] {$\downarrow$};
\end{tikzpicture}
    \caption{The adjacency matrix of the 3-partite graph in
Figure~4, where the rows and columns are arranged according
to the increasing order of their indices, and corresponding
$\mathcal{R}_i$'s and $\mathcal{C}_j$'s are shown in the
elliptic regions. Its interval representation is as follows:
Red: $I_{v_1}=[1,3]$, $I_{v_6}=[6,9]$, $I_{v_{10}}=[10,10]$;
Green: $I_{v_2}=[3,4]$, $I_{v_8}=[8,10]$, $I_{v_9}=[9,9]$;
Black: $I_{v_3}=[3,3]$, $I_{v_4}=[4,6]$, $I_{v_5}=[5,8]$,
$I_{v_7}=[7,9]$.
}
    \label{}
\end{figure}
\noindent{}
Note that the zeros inside the elliptic regions of $\mathcal{R}_{5}$ and $\mathcal{C}_{5}$, as shown in 
Figure~5, do not belong to the sets $\mathcal{R}_{5}$ and $\mathcal{C}_{5}$. They simply indicate that 
the 1’s in $\mathcal{R}_{5}$ and $\mathcal{C}_{5}$ are not strictly consecutive, but rather almost consecutive.
\par Pavol Hell and Jing Huang \cite{hell-huang} characterized interval bigraphs using forbidden patterns with respect to a specific vertex ordering in the following theorem.
\begin{theo}[\cite{hell-huang}]
   \textit{ Let $H$ be a bipartite graph with bipartition $(X,Y)$. Then the following statements are equivalent}:
    \begin{itemize}
        \item\textit{ $H$ is an interval bigraph};
        \item\textit{ the vertices of $H$ can be ordered $v_1$, $v_2$, ..., $v_n$, so that there do not exist $a<b<c$ in the configurations in Figure 6. (Black vertices are in $X$, red vertices in $Y$, or conversely, and all edges not shown are absent.)}
        \item \textit{the vertices of $H$ can be ordered $v_1$, $v_2$, ..., $v_n$, so that there do not exist $a<b<c<d$ in the configurations in Figure 7}.
    \end{itemize}
    
\end{theo}
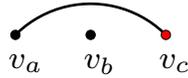
\begin{figure}[H]
    \centering
    \begin{tikzpicture}[line cap=round,line join=round,x=1.0cm,y=1.0cm,scale=.5]
\clip(1.,-1) rectangle (7.,3.);
\draw [shift={(4.,-1.)},line width=1.pt]  plot[domain=0.7853981633974483:2.356194490192345,variable=\t]({1.*2.8284271247461903*cos(\t r)+0.*2.8284271247461903*sin(\t r)},{0.*2.8284271247461903*cos(\t r)+1.*2.8284271247461903*sin(\t r)});
\begin{scriptsize}
\draw [fill=black] (2.,1.) circle (3.5pt);
\draw [fill=black] (4.,1.) circle (3.5pt);
\draw [fill=red] (6.,1.) circle (3.5pt);
\draw (1.5,.8) node[anchor=north west,scale=1.5] {$v_a$};
\draw (3.5,.8) node[anchor=north west,scale=1.5] {$v_b$};
\draw (5.5,.8) node[anchor=north west,scale=1.5] {$v_c$};
\end{scriptsize}
\end{tikzpicture}
    \caption{Forbidden pattern.}
    \label{fig:enter-label-3}
\end{figure}
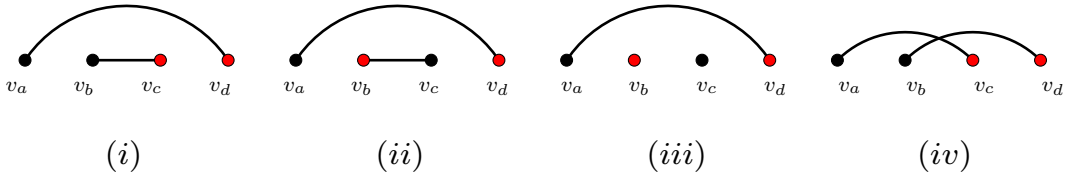
\begin{figure}[H]
    \centering
    \begin{tikzpicture}[line cap=round,line join=round,x=1.0cm,y=1.0cm, scale=.9]
\clip(1.,0.) rectangle (19.,3.);
\draw [line width=1.pt] (3.,2.)-- (4.,2.);
\draw [shift={(3.5,1.)},line width=1.pt]  plot[domain=0.5880026035475675:2.5535900500422257,variable=\t]({1.*1.8027756377319948*cos(\t r)+0.*1.8027756377319948*sin(\t r)},{0.*1.8027756377319948*cos(\t r)+1.*1.8027756377319948*sin(\t r)});
\draw [line width=1.pt] (7.,2.)-- (8.,2.);
\draw [shift={(7.5,1.)},line width=1.pt]  plot[domain=0.5880026035475675:2.5535900500422257,variable=\t]({1.*1.8027756377319948*cos(\t r)+0.*1.8027756377319948*sin(\t r)},{0.*1.8027756377319948*cos(\t r)+1.*1.8027756377319948*sin(\t r)});
\draw [shift={(11.5,1.)},line width=1.pt]  plot[domain=0.5880026035475675:2.5535900500422257,variable=\t]({1.*1.8027756377319948*cos(\t r)+0.*1.8027756377319948*sin(\t r)},{0.*1.8027756377319948*cos(\t r)+1.*1.8027756377319948*sin(\t r)});
\draw [shift={(15.,1.)},line width=1.pt]  plot[domain=0.7853981633974483:2.356194490192345,variable=\t]({1.*1.4142135623730951*cos(\t r)+0.*1.4142135623730951*sin(\t r)},{0.*1.4142135623730951*cos(\t r)+1.*1.4142135623730951*sin(\t r)});
\draw [shift={(16.,1.)},line width=1.pt]  plot[domain=0.7853981633974483:2.356194490192345,variable=\t]({1.*1.4142135623730951*cos(\t r)+0.*1.4142135623730951*sin(\t r)},{0.*1.4142135623730951*cos(\t r)+1.*1.4142135623730951*sin(\t r)});
\begin{scriptsize}
\draw [fill=black] (2.,2.) circle (2.5pt);
\draw [fill=black] (3.,2.) circle (2.5pt);
\draw [fill=red] (4.,2.) circle (2.5pt);
\draw [fill=red] (5.,2.) circle (2.5pt);
\draw [fill=black] (6.,2.) circle (2.5pt);
\draw [fill=red] (7.,2.) circle (2.5pt);
\draw [fill=black] (8.,2.) circle (2.5pt);
\draw [fill=red] (9.,2.) circle (2.5pt);
\draw [fill=black] (10.,2.) circle (2.5pt);
\draw [fill=red] (11.,2.) circle (2.5pt);
\draw [fill=black] (12.,2.) circle (2.5pt);
\draw [fill=red] (13.,2.) circle (2.5pt);
\draw [fill=black] (14.,2.) circle (2.5pt);
\draw [fill=black] (15.,2.) circle (2.5pt);
\draw [fill=red] (16.,2.) circle (2.5pt);
\draw [fill=red] (17.,2.) circle (2.5pt);
\draw (1.6,1.8) node[anchor=north west,scale=1.] {$v_a$};
\draw (2.6,1.8) node[anchor=north west,scale=1.] {$v_b$};
\draw (3.6,1.8) node[anchor=north west,scale=1.] {$v_c$};
\draw (4.6,1.8) node[anchor=north west,scale=1.] {$v_d$};
\draw (3,1) node[anchor=north west,scale=1.5] {$(i)$};

\draw (5.7,1.8) node[anchor=north west,scale=1.] {$v_a$};
\draw (6.7,1.8) node[anchor=north west,scale=1.] {$v_b$};
\draw (7.7,1.8) node[anchor=north west,scale=1.] {$v_c$};
\draw (8.7,1.8) node[anchor=north west,scale=1.] {$v_d$};
\draw (7,1) node[anchor=north west,scale=1.5] {$(ii)$};

\draw (9.8,1.8) node[anchor=north west,scale=1.] {$v_a$};
\draw (10.8,1.8) node[anchor=north west,scale=1.] {$v_b$};
\draw (11.8,1.8) node[anchor=north west,scale=1.] {$v_c$};
\draw (12.8,1.8) node[anchor=north west,scale=1.] {$v_d$};
\draw (11,1) node[anchor=north west,scale=1.5] {$(iii)$};

\draw (13.9,1.8) node[anchor=north west,scale=1.] {$v_a$};
\draw (14.9,1.8) node[anchor=north west,scale=1.] {$v_b$};
\draw (15.9,1.8) node[anchor=north west,scale=1.] {$v_c$};
\draw (16.9,1.8) node[anchor=north west,scale=1.] {$v_d$};
\draw (15,1) node[anchor=north west,scale=1.5] {$(iv)$};

\end{scriptsize}
\end{tikzpicture}
    \caption{Forbidden patterns.}
    \label{fig:enter-label-4}
\end{figure}
It is a natural question whether the classes of interval $r$-graphs ($r \geq 3$) can be characterized by a finite collection of forbidden patterns with respect to some specific ordering of their vertices. The following theorem provides an answer to this question.
\begin{theo}\label{t4}
    Let $G = (X_1, X_2, \ldots, X_r, E)$ be an $r$-partite graph ($r \geq 3$). Then the following 
statements are equivalent:
\begin{enumerate}[i.]
    \item $G$ is an interval $r$-graph;
    \item The vertices of $G$ can be ordered $v_1, v_2, \ldots, v_n$ such that no indices 
    $i < j < k $ occur in the configurations shown in Figure~8. 
    (Here, different colors indicate that vertices belong to 
    different partite sets; moreover, all edges not explicitly drawn are absent.)
\end{enumerate}
\end{theo}
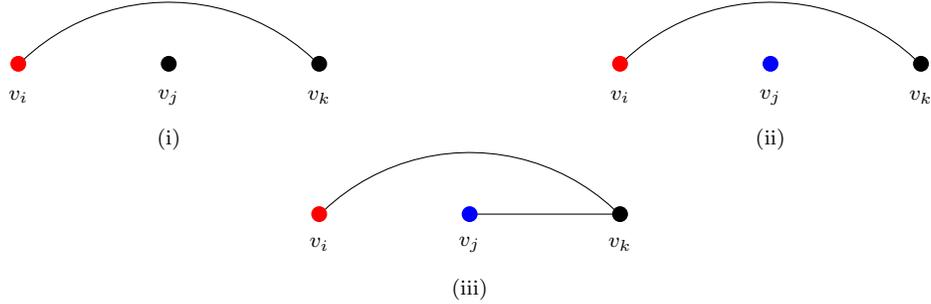
\begin{figure}[H]
    \centering
    \begin{tikzpicture}[line cap=round,line join=round,x=1.0cm,y=1.0cm,scale=2]
\clip(0,-2) rectangle (8,1);
\draw [shift={(2,-1)}] plot[domain=0.79:2.36,variable=\t]({1*1.41*cos(\t r)+0*1.41*sin(\t r)},{0*1.41*cos(\t r)+1*1.41*sin(\t r)});
\draw [shift={(6,-1)}] plot[domain=0.79:2.36,variable=\t]({1*1.41*cos(\t r)+0*1.41*sin(\t r)},{0*1.41*cos(\t r)+1*1.41*sin(\t r)});
\draw [shift={(4,-2)}] plot[domain=0.79:2.36,variable=\t]({1*1.41*cos(\t r)+0*1.41*sin(\t r)},{0*1.41*cos(\t r)+1*1.41*sin(\t r)});
\draw (4,-1)-- (5,-1);
\begin{scriptsize}
\fill [color=red] (1,0) circle (1.5pt);
\draw[color=black] (1.,-0.22) node {$v_i$};
\fill [color=black] (2,0) circle (1.5pt);
\draw[color=black] (2.,-0.22) node {$v_j$};
\fill [color=black] (3,0) circle (1.5pt);
\draw[color=black] (3.,-0.22) node {$v_k$};
\draw[color=black] (2.,-.5) node {(i)};

\fill [color=red] (5,0) circle (1.5pt);
\draw[color=black] (5.,-0.22) node {$v_i$};
\fill [color=blue] (6,0) circle (1.5pt);
\draw[color=black] (6.,-0.22) node {$v_j$};
\fill [color=black] (7,0) circle (1.5pt);
\draw[color=black] (7.,-0.22) node {$v_k$};
\draw[color=black] (6.,-.5) node {(ii)};

\fill [color=red] (3,-1) circle (1.5pt);
\draw[color=black] (3.,-1.2) node {$v_i$};
\fill [color=blue] (4,-1) circle (1.5pt);
\draw[color=black] (4.,-1.2) node {$v_j$};
\fill [color=black] (5,-1) circle (1.5pt);
\draw[color=black] (5.,-1.2) node {$v_k$};
\draw[color=black] (4.,-1.5) node {(iii)};

\end{scriptsize}
\end{tikzpicture}
    \caption{Forbidden patterns for interval $r$-graphs ($r \geq 3$) in terms of specific orderings of their vertices. }
    \label{fig:placeholder}
\end{figure}
\begin{proof}
  \noindent\textbf{Necessity.} 
Let $G = (X_1, X_2, \ldots, X_r, E)$ $(r\geq 3)$ be an \emph{interval $r$-graph}. 
Let $\mathcal{I} = \{ I_v : v \in \bigcup_{i=1}^{r} X_i \}$ denote an interval model of $G$. \\
Without loss of generality, we may assume that the interval model \( \mathcal{I} \) satisfies the following conditions:
\begin{enumerate}
    \item All intervals are closed, i.e., each interval includes its endpoints.
    \item No two intervals share the same left endpoint.
\end{enumerate}

Arrange the vertices of $G$ according to the increasing order of the left endpoints of the corresponding intervals.

\smallskip
\noindent
Consider any three vertices $v_i, v_j, v_k$ of $G$ belonging to three distinct partite sets 
$X_\alpha, X_\beta, X_\gamma$, respectively, where $i < j < k$. 
Suppose that $v_i$ and $v_k$ are adjacent. 
According to the vertex ordering, the left endpoint of the interval $I_{v_k}$ must lie within the interval $I_{v_i}$. 
Consequently, the left end points of the intervals $I_{v_{i+1}}, I_{v_{i+2}}, \ldots, I_{v_{k-1}}$ are contained in $I_{v_i}$, 
and therefore each of them intersects $I_{v_i}$.

\smallskip
\noindent
It follows that every vertex indexed between $v_i$ and $v_k$ that does not belong to $X_\alpha$ must be adjacent to $v_i$. 
In particular, $v_j \in X_\beta$ is adjacent to $v_i$. 
Hence, the configurations (ii) and (iii) shown in Figure~8, cannot occur in an interval $r$-graph.

\smallskip
\noindent
Similarly, if the vertices $v_j$ and $v_k$ belong to the same partite set, 
then the configuration (i) illustrated in Figure~8, also cannot occur in an interval $r$-graph. 
Thus, none of the three configurations shown in Figure~8 can occur in $G$.\\

\textbf{Sufficiency:}  
Suppose \( G = (X_1, X_2, \ldots, X_r, E) \) is an \( r \)-partite graph that admits an ordering \( v_1, v_2, \ldots, v_n \) such that none of the configurations illustrated in Figure~8 occur. We shall show that \( G \) is an interval \( r \)-graph.  

Place the vertices of \( G \) on the real line such that each vertex \( v_i \) is positioned at the positive integer point \( i \), for \( 1 \leq i \leq n \). We now construct an interval corresponding to each vertex \( v_i \) (\( 1 \leq i \leq n \)).  

For a vertex \( v_i \in X_\alpha \), define \( t_i \) to be the largest vertex index (in the rightward direction from \( v_i \)) corresponding to a vertex \( v_{t_i} \in X \setminus X_\alpha \) such that all vertices \( v_l \in X \setminus X_\alpha \), where \( i < l \leq t_i \), are adjacent to \( v_i \). If no such vertex index exists, then set \( t_i = i \).

Now define a family of intervals
\[
\mathcal{I}' = \{ I'_{v_i} : 1 \le i \le n \}
\]
corresponding to the vertex set of \( G \), where
\[
I'_{v_i} = [i, \, t_i + \tfrac{1}{2} ].
\]
We will show that this family of intervals is an interval representation of the graph \( G \).  

\bigskip
\noindent
\textbf{Claim 1.}
If two vertices \( v_i \in X_\alpha \) and \( v_k \in X_\gamma \) (with \( \alpha \ne \gamma \) and \( i < k \)) satisfy  
\( I'_{v_i} \cap I'_{v_k} \ne \emptyset \), then \( v_i \) and \( v_k \) are adjacent in \( G \).

\smallskip
\emph{Proof.}
From the definition of the intervals,
\[
I'_{v_i} = [i, t_i + \tfrac{1}{2}], \quad I'_{v_k} = [k, t_k + \tfrac{1}{2}].
\]
If \( I'_{v_i} \cap I'_{v_k} \ne \emptyset \), then \( k \in I'_{v_i} \), implying
\[
i \le k \le t_i + \tfrac{1}{2} \implies i \le k \le t_i.
\]
Hence, \( v_i \) is adjacent to \( v_k \). 

\bigskip
\noindent
\textbf{Claim 2.}
Conversely, suppose \( v_i \) and \( v_k \) (\( i < k \)) are adjacent in \( G \), where  
\( v_i \in X_\alpha \) and \( v_k \in X_\gamma \).  
We will show that \( I'_{v_i} \cap I'_{v_k} \ne \emptyset \).

\smallskip
\emph{Proof.}
For contradiction, assume that \( I'_{v_i} \cap I'_{v_k} = \emptyset \).  
Then, by definition of the intervals, \( t_i < k \).  
This means that \( v_i \) is not  adjacent to every vertex consecutively (from partite sets different from \( X_\alpha \)) on its right up to \( v_k \).  
Hence, there exists a vertex \( v_j \) such that \( v_j \notin X_\alpha \), \( i < j < k \), and \( v_i \) is not adjacent to \( v_j \).

We consider the following cases:

\begin{itemize}
  \item \textbf{Case 1.} \( v_j \) and \( v_k \) belong to the same partite set $X_\gamma$.  
  Then configuration (i) (see Figure~8) occurs --- a contradiction.

  \item \textbf{Case 2.} \( v_i, v_j, v_k \) belong to three distinct partite sets $X_\alpha, X_\beta, X_\gamma$ respectively.
  \begin{itemize}
    \item \emph{Subcase 1:} If \( v_j \) is not adjacent to \( v_k \), then configuration (ii) (see Figure~8) occurs --- a contradiction.
    \item \emph{Subcase 2:} If \( v_j \) is adjacent to \( v_k \), then configuration (iii) (see Figure~8) occurs --- a contradiction.
  \end{itemize}
\end{itemize}

Thus, in every possible case, we reach a contradiction.  
Therefore, if \( v_i \) and \( v_k \) are adjacent, then \( I'_{v_i} \cap I'_{v_k} \ne \emptyset \).

\bigskip
Hence, the constructed family of intervals \( \mathcal{I}' \) serves as an interval model of the \( r \)-partite graph \( G \). Therefore, \( G \) is an interval \( r \)-graph.

\end{proof}

\section{Conclusion}
In this paper, we discussed various types of vertex orderings for the most generalized version of interval bigraphs, namely interval $r$-graphs. We also identified the corresponding forbidden patterns for this class with respect to certain specific orderings. For future research, one could aim to extend these concepts to the class of circular-arc graphs, with the goal of generalizing the notion of circular-arc bigraphs as well. It would also be interesting to investigate whether a finite set of forbidden patterns exists for these generalized classes.

\noindent{}\textbf{Acknowledgement} \\
The first author sincerely acknowledges the Council of Scientific and Industrial Research (CSIR), India [File No. 09/0028(11986)/2021-EMR-I], for providing financial support through a fellowship to pursue the Ph.D. degree.

\end{document}